%% file: main.tex
\documentclass[11pt,final]{article}

\usepackage{color}
\usepackage[colorlinks=true, allcolors=blue]{hyperref}
\usepackage{amsmath, amssymb, amsthm}
\usepackage{mathrsfs}
\usepackage{mathtools}
\usepackage[noabbrev,capitalize,nameinlink]{cleveref}
\crefname{equation}{}{}
\usepackage{fullpage}
\usepackage[noadjust]{cite}
\usepackage{graphics}
\usepackage{pifont}
\usepackage{tikz}
\usepackage{bbm}
\usepackage[T1]{fontenc}

\usetikzlibrary{arrows.meta}

\usepackage{environ}
\usepackage{framed}
\usepackage{url}
\usepackage[linesnumbered,ruled,vlined]{algorithm2e}
\usepackage[noend]{algpseudocode}
\usepackage[labelfont=bf]{caption}
\usepackage{cite}
\usepackage{framed}
\usepackage[framemethod=tikz]{mdframed}
\usepackage{appendix}
\usepackage{graphicx}
\usepackage[textsize=tiny]{todonotes}
\usepackage{tcolorbox}
\allowdisplaybreaks[1]

\crefname{algocf}{Algorithm}{Algorithms}

\crefname{equation}{}{} %remove ``Equation''
\AtBeginEnvironment{appendices}{\crefalias{section}{appendix}} %appendices

\usepackage{enumerate}
\usepackage{showlabels} %to remove later after labels fixed
\usepackage[color,final]{showkeys} %add in 'final' into parameter to remove showkeys

\colorlet{refkey}{orange!20}
\colorlet{labelkey}{blue!30}

\crefname{algocf}{Algorithm}{Algorithms}

\numberwithin{equation}{section}
\newtheorem{theorem}{Theorem}[section]
\newtheorem{proposition}[theorem]{Proposition}
\newtheorem{lemma}[theorem]{Lemma}

\crefname{claim}{Claim}{Claims}

\newtheorem{corollary}[theorem]{Corollary}

\newtheorem*{question*}{Question}
\newtheorem{fact}[theorem]{Fact}

\theoremstyle{definition}
\newtheorem{definition}[theorem]{Definition}

\newtheorem*{definition*}{Definition}

\theoremstyle{remark}
\newtheorem*{remark}{Remark}

\newcommand{\set}[1]{\bigg\{ #1 \bigg\}}

\newcommand{\mb}{\mathbb}

\newcommand{\mbm}{\mathbbm}
\newcommand{\mc}{\mathcal}

\newcommand{\on}{\operatorname}
\newcommand{\wh}{\widehat}

\newcommand{\var}{\on{Var}}

	\newcommand{\declaredelims}[3]{%
		\expandafter\DeclarePairedDelimiterX\csname#1\endcsname[1]{#2}{#3}{\renewcommand\given{\GivenSymbol[\delimsize]}\renewcommand\mid{\MidSymbol[\delimsize]}##1}%
	}
	
	\newcommand{\declarenameddelims}[4]{%
		\expandafter\DeclarePairedDelimiterXPP\csname#1\endcsname[1]{#2}{#3}{#4}{}{\renewcommand\given{\GivenSymbol[\delimsize]}\renewcommand\mid{\MidSymbol[\delimsize]}##1}%
		\expandafter\DeclarePairedDelimiterXPP\csname#1X\endcsname[2]{#2_{##1}}{#3}{#4}{}{\renewcommand\given{\GivenSymbol[\delimsize]}\renewcommand\mid{\MidSymbol[\delimsize]}##2}%
	}

	\declarenameddelims{P}{\mathbb{P}}\lbrack\rbrack
	\declarenameddelims{Ex}{\mathbb{E}}\lbrack\rbrack

	\let\R\relax
	\newcommand*{\R}{{\mathbb{R}}}

	\let\1\relax
	\newcommand*{\1}{{\mathds{1}}}
	\let\l\relax
	\newcommand*{\l}{{\ell}}
	
	\let\poly\relax
	\DeclareMathOperator{\poly}{poly}

	\providecommand{\given}{}
	\newcommand\GivenSymbol[1][]{%
		\nonscript\:#1\vert%
		\allowbreak%
		\nonscript\:%
		\mathopen{}%
	}

\allowdisplaybreaks

\author{
Vishesh Jain \\
Stanford University \\
\texttt{visheshj@stanford.edu}
\and
Huy Tuan Pham \\
Stanford University \\
\texttt{huypham@stanford.edu}
\and
Thuy Duong Vuong \\
Stanford University \\
\texttt{tdvuong@stanford.edu}
}

\date{}

\begin{document}
\title{Spectral independence, coupling with the stationary distribution, and the spectral gap of the Glauber dynamics}

\begin{titlepage}
\clearpage\maketitle
\thispagestyle{empty}
\input{abstract.tex}
\end{titlepage}
\newpage

\section{Introduction}
Consider an undirected graph $G = (V,E)$ with vertices $V$ and edges $E \subseteq \binom{V}{2}$. Let $q \geq 2$ be an integer and $[q]$ denote the discrete interval $\{1,\dots, q\}$. A \emph{$q$-spin system} on the graph $G = (V,E)$ is parameterized by an entry-wise non-negative symmetric matrix $A \in \mb{R}^{q \times q}_{\geq 0}$ (the `interaction matrix') and an entry-wise positive vector $h \in \mb{R}^{q}_{>0}$ (the vector of `external fields'). This definition includes many widely studied objects in statistical physics, theoretical computer science, and combinatorics. We list three examples, which will be revisited when we discuss applications of the main result of this article. 

\begin{itemize}
    \item Zero-temperature antiferromagnetic Potts model. Here, $A = J_{q \times q} - I_{q \times q}$, where $I_{q \times q}$ is the $q\times q$ identity matrix and $J_{q \times q}$ is the $q\times q$ all-ones matrix, and $h = \mbm{1}_{q}$, the $q$-dimensional all-ones vector. The Gibbs distribution corresponds to the uniform distribution on proper $q$-colorings of $G = (V,E)$.
    
    \item Hardcore model. Here, $q = 2$ (it is conventional to identify $[2]$ with $\{0,1\}$ in this case), $h = \lambda \mbm{1}_{2}$, where $\lambda > 0$ is known as the `fugacity' and $\mbm{1}_{2}$ is the $2$-dimensional all-ones vector, and 
    \[A = 
    \begin{bmatrix}
    1 & 1\\
    1 & 0
    \end{bmatrix}
    \]
    The Gibbs distribution corresponds to a distribution over independent sets of $G = (V,E)$ where the probability of an independent set $I$ is proportional to $\lambda^{|I|}$.
    
    \item Monomer-dimer model. For a graph $G = (V,E)$, recall that the line graph $L(G) = (E, \wh{E})$ is a graph with vertices $E$ and for $e \neq e' \in E$, an edge $\{e,e'\} \in \wh{E}$ if and only if $e,e'$ share a vertex in $G$. The monomer-dimer model on $G = (V,E)$ with parameter $\lambda > 0$ refers to the hardcore model on $L(G)$ with fugacity $\lambda$. The Gibbs distribution is the distribution on matchings of $G$ where the probability of a matching $M$ is proportional to $\lambda^{|M|}$.
\end{itemize}

A configuration of the spin system is an assignment of spins to vertices, i.e.~an element $\sigma \in [q]^{V}$. The Gibbs distribution $\mu = \mu_{G,A,h}$ is a probability distribution on $[q]^{V}$ defined, for $\sigma \in [q]^{V}$, by 
\[\mu(\sigma) = \frac{1}{Z_G(A,h)}\prod_{\{u,v\}\in E}A(\sigma_u, \sigma_v) \prod_{v \in V}h(\sigma_v),\]
where the normalizing constant
\[Z_G(A,h) = \sum_{\sigma \in [q]^V}\prod_{\{u,v\}\in E}A(\sigma_u, \sigma_v) \prod_{v \in V}h(\sigma_v) \]
is known as the partition function.

Sampling from the Gibbs distribution and approximating the partition function are fundamental computational tasks \cite{jerrum2003counting,LP17}. The Markov Chain Monte Carlo (MCMC) paradigm (cf.~\cite{aldous2002reversible, jerrum2003counting, LP17, sinclair1989approximate}) provides perhaps the most versatile and powerful approach to these very general problems and has been the subject of intense study in the past four decades. A particularly simple and popular Markov chain for sampling from the Gibbs distribution of a $q$-spin system on a graph $G = (V,E)$ is the \emph{(single site) Glauber dynamics (or Gibbs sampling)}, defined as follows: starting from a (possibly random) initial configuration $X_0 \in [q]^{V}$, for each integer $t\geq 1$, the configuration $X_{t}$ is generated from $X_{t-1}$ as follows: let $v$ be a uniformly chosen element of $V$ and let $Q$ be sampled from the distribution
\[\mu[\sigma_{v} = \cdot \mid \sigma_{w} = (X_{t-1})(w) \text{      } \forall w\neq v].\]
Then, set $X_{t}(w) = X_{t-1}(w)$ for all $w\neq v$ and $X_{t}(v) = Q$.

Let $P_{\mu}$ denote the transition matrix of the Glauber dynamics. It is readily seen that $P_{\mu}$ is reversible with respect to $\mu$ i.e.~for all $\sigma, \sigma' \in [q]^{V}$,
\[\mu(\sigma)P_{\mu}(\sigma, \sigma') = \mu(\sigma')P_{\mu}(\sigma',\sigma).\]
In particular, $\mu$ is a stationary distribution for $P_{\mu}$. Assuming further that $P_{\mu}$ is irreducible (this will be readily satisfied in all our applications), $\mu$ is the unique stationary distribution for $P_{\mu}$. Since $P_{\mu}$ is trivially aperiodic, it follows (cf.~\cite{LP17}) that in this case, for any distribution on the starting configuration $X_0$,
\[\limsup_{t\to \infty}\|P_{\mu}^{t}(X_0, \cdot) - \mu\|_{\on{TV}} = 0,\]
where $\|\cdot - \cdot\|_{\on{TV}}$ denotes the total variation distance between probability distributions.  

For algorithmic applications, we are interested in the rate at which $\|P_{\mu}^{t}(X_0, \cdot) - \mu\|_{\on{TV}}$ decays to $0$. In the worst-case scenario, this is captured by the \emph{mixing time}. Concretely, for an ergodic transition matrix $P$ on a finite state space $\mc{S}$ with stationary distribution $\mu$, and for $\varepsilon \in (0,1)$, the $\varepsilon$-mixing time is defined as
\[\tau_{\on{mix}}(\varepsilon) = \max_{\sigma \in \mc{S}}\min\{t\geq 0: \|P^{t}(\sigma, \cdot) - \mu\|_{\on{TV}} \leq \varepsilon\}.\]
The term mixing time commonly refers to $\tau_{\on{mix}}(1/4)$.

The previous definition considers the worst-case starting state. If the initial state $X_0$ is distributed according to the probability distribution $\pi_0$ on $\mc{S}$, we have the more refined quantity
\[\tau_{\on{mix}}(\varepsilon, \pi_0) = \min\{t\geq 0: \|P^{t}(\pi_0, \cdot) - \mu\|_{\on{TV}} \leq \varepsilon\}.\]

Classical methods for bounding the mixing time include the method of canonical paths (\cite{jerrum1989approximating}) and the coupling method (cf.~\cite{LP17}). The past few years have witnessed the emergence of an attractive method for bounding the mixing time, based on local-to-global arguments for high-dimensional expanders \cite{ALO20, dinur2017high, kaufman2018high, alev2020improved, oppenheim2018local}. Of direct relevance to us is the work of Anari, Liu, and Oveis Gharan \cite{ALO20}, who introduced the notion of \emph{spectral independence} (see \cref{sub:spectral-independence} for an introduction) as a way of proving that the Glauber dynamics mixes rapidly. This notion, introduced in \cite{ALO20} for Boolean spin systems, was further developed in the works \cite{feng2021rapid, chen2021rapid}. We defer a precise definition to \cref{def:spectral-independence} in \cref{sub:spectral-independence}, but the upshot is the following:
\begin{itemize}
    \item Many interesting spin systems, such as the hardcore model below critical fugacity (\cite{ALO20}), proper colorings of triangle free graphs of maximum degree $\Delta$ with $(1.763.. + \delta)\Delta$ colors (\cite{feng2021rapid, chen2021rapid}, and antiferromagnetic 2-spin systems on bounded degree graphs in the tree uniqueness regime (with some gap, \cite{chen2020rapid, CLV20}) are $(C,\eta)$-spectrally independent with $C = O(1)$ and $\eta \in [0,1)$, possibly close to $1$ inverse polynomially in $q$ or the maximum degree $\Delta$.
    \item (\cite{ALO20}, extended by \cite{feng2021rapid, chen2021rapid}; see also \cref{thm:ALO} below) For a $(C,\eta)$-spectrally independent $q$-spin system on a graph $G = (V,E)$, the Glauber dynamics mixes in time 
    $$O(|V|^{2+2C}\cdot (1-\eta)^{-2-2C}\cdot \log{q}).$$
\end{itemize}
While this approach was successful in providing the first polynomial time approximate sampling algorithms for many interesting models, the drawback is that the parameter $C$ can be quite a large constant; for instance, in the case of proper colorings of triangle free graphs of maximum degree $\Delta$ with $(1.763.. + \delta)\Delta$ colors, the best known bound \cite{feng2021rapid, chen2021rapid} is $C = O(1/\delta)$, which leads to a mixing time of the form $$O\left(|V|^{O(1/\delta)}\right).$$

In the case of $q$-spin systems on graphs $G = (V,E)$ of maximum degree at most $\Delta$ and for which the marginals of the Gibbs distribution are lower bounded by $b > 0$ (even under conditioning on an arbitrary proper subset of the spins), a remarkable recent paper of Chen, Liu, and Vigoda \cite{CLV20} showed that the mixing time of the Glauber dynamics is
\[O\left(n\log{n} \cdot \left(\frac{\Delta}{b}\right)^{O\left(\frac{C}{b^{2}}\right)} \right),\]
which was improved in a very recent work \cite{blanca2021mixing} to
\[O\left(n\log{n} \cdot \left(\frac{\Delta}{b}\right)^{O\left(1 + \frac{C}{b}\right)} \right).\]
For $q$-spin systems where $\Delta, q, b^{-1} = O(1)$ (these conditions are guaranteed by boundedness of $\Delta,q$ as well as of the entries of the interaction matrix and external field), the dependence on $n$ is optimal \cite{hayes2005general}. However, the running time grows quite rapidly with the parameters $\Delta, q, b^{-1}$ -- for instance, in the case of properly coloring triangle free graphs of maximum degree $\Delta$ with $(1.763\dots + \delta)\Delta$ colors, the running time is of the form
\[O\left(n\log{n}\cdot \Delta^{O(\delta^{-1} \cdot \Delta)}\right).\]

It was asked in \cite[Section~8]{CLV20} whether the dependence of the running time on the maximum degree and the spectral independence parameters can be improved. This is the focus of the present work.

\subsection{Our results} Our results are best stated in terms of the \emph{spectral gap} of the Glauber dynamics. Let $\mu$ denote the Gibbs distribution of a $q$-spin system on $G = (V,E)$ and let $P$ denote the transition matrix of the Glauber dynamics on the state space $\mc{S} = [q]^{V}$. Since $P$ is reversible with respect to $\mu$, all the eigenvalues of $P$ are real. Let us denote these eigenvalues by
\[1 = \lambda_{1} \geq \lambda_{2} \geq \dots \geq \lambda_{|\mc{S}|} \geq -1\]
The (absolute) spectral gap of $P$ is defined by
\[\lambda_{\ast} = 1 - \max\{|\lambda_2|, |\lambda_{|\mc{S}|}|\}\]
and the relaxation time is defined by
\[\tau_{\on{rel}} = \frac{1}{\gamma_\ast}.\]
The following relations between the $\varepsilon$-mixing time and the relaxation time for a reversible, ergodic transition matrix $P$ is well-known (cf.~\cite[Equation~12.8, Theorem~12.5]{LP17}): for all $\varepsilon \in (0,1)$ and for all probability distributions $\pi_0$ on $\mc{S}$ 
\begin{equation}
\label{eqn:gap-to-warm}
    \tau_{\on{mix}}(\varepsilon, \pi_0) \leq \tau_{\on{rel}}\log\left(\frac{1}{\varepsilon}\cdot \max_{x\in \mc{S}}\frac{\pi_0(x)}{\mu(x)}\right).
\end{equation}
In particular,
\begin{equation}
    \label{eqn:gap-to-mixing}
    \tau_{\on{mix}}(\varepsilon) \leq \tau_{\on{rel}}\log\left(\frac{1}{\varepsilon}\cdot \frac{1}{\min_{x\in \mc{S}}\mu(x)}\right).
\end{equation}
We also have a bound in the other direction:
\begin{equation}
    \label{eqn:mixing-to-gap}
    \tau_{\on{mix}}(\varepsilon) \geq (\tau_{\on{rel}}-1)\log\left(\frac{1}{2\varepsilon}\right).
\end{equation}

Our main technical result is the following. 
\begin{theorem}\label{thm:spectral gap}
Consider a graph $G = (V,E)$ of maximum degree $\Delta$. Let $C \geq 0$ and $\eta \in [0,1)$. Suppose that $\mu$ is a $(C,\eta)$-spectrally independent distribution on $[q]^{V}$ and that $50\lceil 2C \rceil \Delta \leq n$. Then, the Glauber dynamics for $\mu$ has spectral gap at least 
\[c_{\ref{thm:spectral gap}} \frac{(1-\eta)^{1+2C}}{(25\Delta \lceil 2C\rceil)^{5\cdot \lceil 2C \rceil}}\cdot \frac{1}{|V|},\]
where $c_{\ref{thm:spectral gap}} > 0$ is a universal constant. 
\end{theorem}

\begin{remark}
This theorem essentially replaces the denominator $|V|^{O(C)}$ appearing in \cref{thm:ALO} of \cite{ALO20, feng2021rapid, chen2021rapid} by the potentially much smaller $\Delta^{O(C)}$. Compared to the results in \cite{CLV20}, the above theorem has the advantage of not requiring any lower bound on the marginals and moreover, has a dependence of $\Delta^{O(C)}$ as opposed to the worse dependence of $\Delta^{O(C\cdot \Delta)}$. On the other hand, we are only able to provide a bound on the spectral gap, whereas \cite{CLV20} provide such bounds even for the modified log-Sobolev constant, which leads to better dependence on $n$ for the mixing time starting from a worst-case initial state.  
\end{remark}

At this point, the improvement in \cref{thm:spectral gap} appears fairly technical, and the reader may rightly wonder if there are any serious applications. The power of \cref{thm:spectral gap} is most apparent when used in combination with well-known coupling arguments. The high-level idea is the following: consider a $q$-spin system on a graph $G = (V,E)$ with maximum degree $\Delta$. Then, \cref{thm:spectral gap} allows us to obtain essentially the correct spectral gap provided that $\Delta = n^{o(1/C)} = n^{o(1)}$, assuming that $C = O(1)$. On the other hand, in many interesting cases, coupling arguments based on `local uniformity' (cf.~\cite{HV06, hayes2003non, hayes2013local, lau2006randomly, dyer2004randomly}) suffice to handle the case when $\Delta$ is at least polylogarithmic in $n$. Taken together, these arguments cover the entire range of $\Delta$. Note that, for such an argument to work, the dependence of the form $O(\Delta^{O(C\cdot \Delta)})$ obtained in \cite{blanca2021mixing} is insufficient since it only works up to $\Delta = o(\log{n}/\log\log{n})$, thereby leaving a gap in the regime for $\Delta$. 

As a concrete illustration of this general idea, by combining \cref{thm:spectral gap} with known spectral independence calculations \cite{feng2021rapid, chen2021rapid}, together with a coupling argument of Hayes and Vigoda \cite{HV06}, we obtain the following. 
\begin{theorem}
\label{thm:coloring}
Let $\alpha^* = 1.763\dots$ denote the unique solution to the equation $x = \exp(1/x)$. Let $G = (V,E)$ be a triangle-free graph of maximum degree at most $\Delta$. Let $|V|=n$. Then, for every $\Delta \geq 3$, for every $\delta > 0$, and for every $k \geq (1+\delta)\alpha^* \Delta$, the Glauber dynamics for the uniform distribution on proper $k$-colorings has spectral gap at least
\[\tilde{\Omega}_{\delta}\left(\frac{1}{n}\right),\]
where $\tilde{\Omega}$ conceals a factor of $e^{O((\log\log{n})^{2})}$.
\end{theorem}

\begin{remark}
From this spectral gap, one can immediately deduce that the mixing time of the Glauber dynamics is $\tilde{O}_{\delta}(n^{2})$ from an arbitrary initialization and $\tilde{O}_{\delta}(n)$ from a ``warm start'' i.e.~the ratio of the distribution of $X_0$ and the Gibbs distribution is polynomially bounded. We note that in previous work \cite{dyer2004randomly}, an optimal bound of $O_{\delta}(n\log{n})$ on the mixing time from an arbitrary initialization was obtained if the girth is at least $5$ and if $\Delta$ is sufficiently large.
\end{remark}

As was mentioned in the remark following \cref{thm:spectral gap}, bounding the modified log-Sobolev constant, as in \cite{CLV20}, leads to running times from worst-case initializations that are out of the reach of analyses based only on the spectral gap. For instance, even an essentially optimal spectral gap bound of $\tilde{\Omega}_{\delta}(n^{-1})$ in \cref{thm:coloring} leads to the sub-optimal mixing time of $\tilde{O}_{\delta}(n^{2})$ from a worst-case initialization. Nonetheless, for applications to \emph{approximate counting}, where one requires many samples, nearly all of which are from a ``warm-start'', this gain of a factor of $n$ from bounding the modified log-Sobolev constant often disappears (see \cite{vstefankovivc2009adaptive}). 

As an application of this, we consider the problem of approximating the partition function on a $\Delta$-regular graph at fugacity $\lambda \leq (1-\delta)\lambda_{c}(\Delta)$, where
$$\lambda_{c}(\Delta) = \frac{(\Delta - 1)^{\Delta-1}}{(\Delta-2)^{\Delta}}$$
is the critical point for the uniqueness/non-uniqueness phase transition on the $\Delta$-regular tree. By using \cref{thm:spectral gap} with known spectral independence calculations \cite{chen2020rapid}, together with a coupling argument of Hayes and Vigoda \cite{HV06}, we are able to quickly recover a (slightly more general version of a) result of Efthymiou et al.~\cite{efthymiou2019convergence}, which had been obtained by using a rather involved and lengthy local-uniformity argument. 

\begin{theorem}
\label{thm:hardcore} 
Let $G = (V,E)$ be a $\Delta$-regular graph with $\Delta \geq 3$ and with girth $\geq 6$, let $\delta > 0$, and let $\lambda \leq (1-\delta) \lambda_c(\Delta).$ Let $Z_{G, \lambda}$ denote the partition function of the hardcore distribution on $G$ with fugacity $\lambda$. Let $|V|=n$. Then the mixing time of the Glauber dynamics for the hardcore distribution at fugacity $\lambda$ from a warm start is
\[\tilde{O}_{\delta}\left(n\right).\]
Moreover, there exists an algorithm which, given $\varepsilon >0,$ outputs (with constant probability) a $(1+\varepsilon)$-multiplicative approximation of $Z_{G, \lambda}$ in time $$\tilde{O}_{\delta}(n^2 \poly(1/\varepsilon)).$$ 
Here, the tilde conceals a factor of $e^{O((\log\log{n})^{2})}$.
\end{theorem}

\begin{remark}
The restriction to $\Delta$-regular graphs (as opposed to graphs of maximum degree $\Delta$) in the above theorem is due to the black-box invocation of a result of Hayes and Vigoda \cite{HV06} and can be likely removed by more careful analysis. The running time $\tilde{O}_{\delta}(n^{2}\on{poly}(1/\varepsilon))$ matches the dependence in $n$ obtained in \cite{efthymiou2019convergence} (for the slightly more restrictive lower bound of $7$ on the girth and for sufficiently large $\Delta$) although note that for obtaining a single sample, their work gives the optimal running time $\tilde{O}_{\delta}(n)$ from an arbitrary initialization. 
\end{remark}

Finally, we note that, even in the absence of an accompanying result in the high maximum-degree regime and even for spectral independence parameters which grow with $n$, \cref{thm:spectral gap} leads to results that may be of interest, for instance, in the study of algorithms on constant average-degree Erd\H{o}s-R\'enyi graphs (which have maximum degree $\Theta(\log{n}/\log\log{n})$). As an example, by combining known spectral independence calculations \cite{CLV20} with \cref{thm:spectral gap}, we can deduce the following. 

\begin{theorem}
\label{thm:matching}
Let $G = (V,E)$ be a graph with maximum degree $\Delta$. Let $|V| = n$ and $|E| = m$. If $\Delta = o((\log n)^2/(\log \log n)^2)$, then the Glauber dynamics for the monomer-dimer model with fugacity $\lambda=1$ (i.e.~the uniform distribution on matchings) mixes in time $O(n^{1+o(1)}m)$ from an arbitrary initial configuration and $O(m^{1+o(1)})$ from a warm-start. 
\end{theorem}

\begin{remark}
For a general graph $G = (V,E)$, the best-known mixing time bound is $\tilde{O}(n^{2}m)$ due to Jerrum and Sinclair \cite{jerrum1989approximating}, refined in \cite{jerrum2003counting}. On the other hand, \cite[Theorem~1.5]{CLV20} along with the improvement in \cite{blanca2021mixing} gives the optimal mixing time $O(m^{1+o(1)})$ for graphs with maximum degree $\Delta = O((\log{n})^{2/3 - o(1)})$. Note that this latter degree bound excludes the case of constant average-degree Erd\H{o}s-R\'enyi graphs.
\end{remark}

\subsection{Concluding remarks and future directions} We have provided a new lower bound on the spectral gap of the Glauber dynamics for spectrally independent spin systems, which is substantially better than existing bounds \cite{ALO20, feng2021rapid, chen2021rapid, CLV20} for many interesting parameter regimes. Notably, in the case of well-studied spin systems such as uniform $k$-colorings of triangle-free graphs and the hardcore model on high girth graphs, where coupling methods have succeeded in analysing the high-degree regime, our bound covers the entire regime outside the scope of the coupling method. In particular, this obviates the need for technical and involved local-uniformity based analyses (at least if one is willing to pay an additional factor of $n$ in the worst-case mixing time), which anyway have introduced additional slack in various girth conditions (\cite{dyer2004randomly,efthymiou2019convergence}).  

A natural direction for future work is therefore to (i) bound the spectral independence of the Gibbs distribution in scenarios where coupling methods have succeeded in the high-degree regime, most notably, for the problem of uniformly sampling $k$-colorings on graphs under girth constraints (\cite{dyer2004randomly, hayes2003non, lau2006randomly}), and (ii) devise (coupling-based) arguments in the high-degree regime for models where spectral independence is known, most notably the hardcore model, even on triangle-free graphs.

% there is $C>0$ such that the Glauber dynamics to sample $k$-colorings of $G$ has spectral gap at least $\frac{1}{\tilde{O}(n)}$, and mixing time satisfies 
% \[\tmix = \begin{cases} O(n \log n) &\text{ when } \Delta \geq C\log n\\ O(n^2 (\log n)^{O(1/\delta^4)}) &\text{ when } \Delta < C\log n\end{cases}\]

\section{Preliminaries}
\subsection{Down-up random walk} As in recent works (cf.~\cite{CLV20, ALO20, chen2020rapid, chen2021rapid, feng2021rapid}) we will find it helpful to view the Glauber dynamics as a `local' walk on a certain weighted simplicial complex. To this end, we record the following definition. 
\begin{definition}[Down-Up Random Walk]\label{def:local-walk} Let $0 \leq \ell \leq k \leq n$ be integers. 
	For a density $\mu:\binom{[n]}{k}\to\R_{\geq 0}$, we define the $k\leftrightarrow\l$ down-up random walk as the sequence of random sets $S_0, S_1,\dots$ generated by the following algorithm:
	\begin{algorithm}
		\For{$t=0,1,\dots$}{
			Select $T_t$ uniformly at random from subsets of size $\l$ of $S_t$.\;
			Select $S_{t+1}$ with probability $\propto \mu(S_{t+1})$ from supersets of size $k$ of $T_t$.
		}
	\end{algorithm}
\end{definition}
In particular,  the Glauber dynamics for the Gibbs distribution of a $q$-spin system on a graph $G = (V, E)$ with $|V| = n$ may be viewed as the $n \leftrightarrow (n-1)$ down-up walk with respect to the distribution $\mu': \binom{V \times [q] }{n } \to \R_{\geq 0}$ defined as follows.
Fix an enumeration $v_1, \dots, v_n$ of $V$. The distribution $\mu'$ is supported on the $n$-element sets $\{(v_1, c_1), \dots, (v_n, c_n)\}$, with $c_1, \dots, c_n \in [q]$, and
\[\mu'(\{(v_1, c_1), \dots, (v_n, c_n) \} ) = \mu(\sigma ), \]
where $\sigma = (c_1,\dots, c_n) \in [q]^{V}$ and $\mu$ denotes the Gibbs distribution of the $q$-spin system.

\subsection{Spectral independence}
\label{sub:spectral-independence}
In this subsection, we record the notion of spectral independence, formalized by \cite{ALO20} in the Boolean setting and further developed in subsequent works \cite{feng2021rapid, chen2021rapid, CLV20}. 

Consider a $q$-spin system on the graph $G = (V,E)$ with Gibbs distribution $\mu$. A configuration $\sigma \in [q]^{V}$ is said to be feasible with respect to $\mu$ if $\mu(\sigma) > 0$ i.e.~if $\sigma$ lies in the support of the measure $\mu$. We will use $\Omega(\mu)$ (or simply $\Omega$, when $\mu$ is clear from context) to denote the set of all feasible configurations with respect to $\mu$. Furthermore, for
$\Lambda \subseteq V$, let $\mu_{\Lambda}$ denote the measure on $[q]^{\Lambda}$ induced by the Gibbs distribution $\mu$ and let 
$$\Omega_{\Lambda} = \set{\tau \in [q]^{\Lambda} : \mu_{\Lambda}(\tau ) > 0},$$
i.e.~$\Omega_{\Lambda}$ denotes the collection of all feasible (partial) configurations
on $\Lambda$. For lightness of notation, we will denote $\Omega_{\{v\}}$ simply by $\Omega_{v}$.
Observe that $\Omega_V = \Omega$. 

For
any subset $\Lambda \subseteq V$ and `boundary condition' $\tau \in \Omega_{\Lambda}$, we will consider the conditional distribution
$\mu_S^{\tau}(\cdot) = \mu(\cdot \mid \sigma_{\Lambda} = \tau )$ over configurations on $S = V \setminus \Lambda$, and we shall write $\Omega_U^\tau$
for the set of feasible
(partial) configurations on $U \subseteq S$ under this conditional measure.
\begin{definition}[Influence Matrix]
Given $\Lambda \subsetneq V $ and $\tau \in \Omega_{\Lambda}$, let 
\[\tilde{V}_{\tau} = \set{(u,i) : u \in V \setminus \Lambda, i \in\Omega_u^{\tau} }\]
For every $(u,i), (v,j)\in \tilde{V}_{\tau} $ with $u\neq v,$ we define the (pairwise) influence of $(u, i)$ on $(v,j)$, conditioned on $\tau$, by
\[\Psi_{\mu}^{\tau} ((u,i), (v,j))  = \mu(\sigma_v = j \mid \sigma_u = i , \sigma_{\Lambda} = \tau ) -  \mu(\sigma_v = j \mid \sigma_{\Lambda} = \tau ) \]
We also set $\Psi_{\mu}^{\tau} ((v,j), (v,j)) =0 $ for all $(v, i),(v, j) \in \tilde{V}_{\tau}$. 

\noindent We call $\Psi_{\mu}^{\tau}$ the (pairwise) influence
matrix conditioned on $\tau.$
\end{definition}

\begin{definition}[Spectral Independence]
\label{def:spectral-independence}
Let $|V| = n$. For parameters $\eta_0 \geq 0, \dots, \eta_{n-2} \geq 0$, we say that a distribution $\mu$ over $[q]^V$ is $(\eta_0, \eta_1,\dots, \eta_{n-2})$-spectrally independent if for every $\Lambda \subseteq V$ with $|\Lambda| \leq n-2$ and for every $\tau \in \Omega_\Lambda$, the largest eigenvalue $\lambda_{\on{max}}(\Psi_\mu^\tau)$ of the influence matrix $\Psi_\mu^\tau$ satisfies $\lambda_{\on{max}}(\Psi_\mu^\tau)\le \eta_{|\Lambda|}$. 
\end{definition}

% \begin{definition}
% we say that a distribution $\mu$ over $[q]^V$ is $(\eta_0, \dots, \eta_{|V|-2})$-spectrally independent if for every $\Lambda \subsetneq V$ and $\tau \in \Omega_\Lambda$, the largest eigenvalue $\lambda_{\on{max}}(\Psi_\mu^\tau)$ of the influence matrix $\Psi_\mu^\tau$ satisfies $\lambda_{\on{max}}(\Psi_\mu^\tau)\le \eta_{|\Lambda|}.$
% \end{definition}

In many applications, such as the ones considered in this article, we can work with the following version of spectral independence requiring fewer parameters. 

\begin{definition}[$(C,\eta)$ spectral independence]
Let $C \geq 0$ and $0 \leq \eta < 1$.
We say that a distribution $\mu$ over $[q]^V$ is $(C, \eta)$-spectrally independent if for every $\Lambda \subseteq V$ with $k:= |\Lambda| \leq |V|-2$ and for every $\tau \in \Omega_\Lambda$, the largest eigenvalue $\lambda_{\on{max}}(\Psi_\mu^\tau)$ of the influence matrix $\Psi_\mu^\tau$ satisfies $$\lambda_{\on{max}}(\Psi_\mu^\tau)\le \min\set{C, \eta(n-k-1)}.$$ 

In other words, $\mu$ is $(\eta_0, \dots , \eta_{n-2})$-spectrally independent with $\eta_i \leq \min\{C, \eta(n-i-1)\}$.
\end{definition}

The choice of the parameterization in the previous definition is explained by the following result of \cite{ALO20}, which lower bounds the spectral gap of the down-up walk with respect to a distribution in terms of the spectral independence of the distribution. 

\begin{theorem}[{\cite[Theorem~1.3]{ALO20}}, {\cite[Theorem~3.2]{feng2021rapid}}. {cf.~\cite[Theorem~6]{chen2021rapid}}]
\label{thm:ALO}
Consider an $(\eta_0,\dots, \eta_{|V|-2})$-spectrally independent distribution $\mu$ on $[q]^{V}$. Then, the spectral gap of the $|V| \leftrightarrow (|V|-1)$ down-up random walk is at least
\[\frac{1}{|V|}\prod_{i=0}^{|V|-2}\left(1- \frac{\eta_i}{n-i-1}\right).\]
In particular, if $\mu$ is $(C,\eta)$-spectrally independent for $C \geq 0$ and $\eta \in [0,1)$, then the spectral gap of the $|V|\leftrightarrow (|V|-1)$ down-up random walk is at least
\[\frac{(1-\eta)^{2+2C}}{|V|^{2C}} \cdot \frac{1}{|V|}\]
\end{theorem}

%\subsection{Markov chain and Mixing time}
% \begin{theorem}[Spectral gap vs. Mixing time {\cite[Thm. 12.4, 12.5]{}}]

% \end{theorem}
%\subsection{Markov chain and Mixing Time}
% TODO: Definition of Glauber dynamics

% TODO: Definition of spectral gap

% TODO: spectral gap vs. mixing time

\begin{comment}
we consider a Markov chain as a triple $(\Omega, P, \pi)$ where $\Omega$ denotes a (finite) state space,
$P \in \R^{\Omega\times \Omega}_{\geq 0}$ denotes a transition probability matrix and $\pi \in \R_{\geq 0}^{\Omega}$ denotes a stationary distribution of the chain
(which will be unique for all chains we consider). For $\tau, \sigma \in \Omega$, we use $P(\tau, \sigma)$ to denote the corresponding entry of $P$, which is the probability of moving from $\tau$ to $\sigma.$
\end{comment}
% We say the Markov chain is reversible if for all pair of states $\tau, \sigma$
% \[\pi(\tau) P(\tau, \sigma) = \pi(\sigma) P(\sigma, \tau)\]
% Note that being reversible means that the transition matrix $P$ is
% self-adjoint w.r.t. the inner product $\langle \cdot, \cdot \rangle_{\pi}$ defined by
% \[\langle f ,g \rangle_{\pi} = \sum_{x\in \Omega} f(x) g(x) \pi(x)\]
%This in turn implies that the eigenvalues of $P$ are real. 
\section{Spectral gap of the Glauber dynamics via spectral independence}
Let $G = G(V,E)$ be a graph of maximum degree $\Delta,$ and $\mu$ be the Gibbs distribution of some $q$-spin system on $G$. Suppose further that $\mu$ is $(C,\eta)$-spectrally independent for some $C \geq 0$ and $\eta \in [0,1)$. In this section, we show how to prove \cref{thm:spectral gap}, which improves the lower bound of \cref{thm:ALO} so that, essentially, the dependence on $|V|$ in the denominator is replaced by similar dependence on $\Delta$. 

Our proof broadly follows the proof of the variance analog of \cite[Theorem~1.9]{CLV20} (see \cite[Appendix~A]{CLV20}). The key difference is the incorporation of \cref{thm:ALO} as an `initial estimate' on the spectral gap, which is then improved by using the general machinery of block factorization of the variance, and the comparison of $\ell$-uniform block factorization of the variance with $1$-uniform block factorization of the variance -- using this `initial estimate' dispenses with the need to assume a lower bound on the (conditional) marginal distributions, as well as leads to the crucial quantitative improvement of the spectral gap underpinning all our applications.\\

% More precisely, we will show the following. 

% \begin{theorem}\label{thm:spectral gap}
% Consider a graph $G = (V,E)$ of maximum degree $\Delta$. Let $C \geq 0$ and $\eta \in [0,1)$. Suppose that $\mu$ is a $(C,\eta)$-spectrally independent distribution on $[q]^{V}$ and that $50\lceil 2C \rceil \Delta \leq n$. Then, the Glauber dynamics for $\mu$ has spectral gap at least 
% \[c_{\ref{thm:spectral gap}} \frac{(1-\eta)^{1+2C}}{(25\Delta \lceil 2C\rceil)^{5\cdot \lceil 2C \rceil}}\cdot \frac{1}{|V|},\]
% where $c_{\ref{thm:spectral gap}} > 0$ is a universal constant. 
% \end{theorem}

The proof requires a few intermediate steps and is presented at the end of this section.

\subsection{Block factorization of variance} 
Recall the notation $\Omega$ and $\Omega_{\Lambda}$ for $\Lambda \subseteq V$. For each $f:\Omega \to \mathbb{R}_{\ge 0}$, $S\subseteq V$, $\tau \in \Omega_{V\setminus S}$, define $\var_S^\tau(f)$ to be the variance of $f$, viewed as a function on $[q]^{S}$, with respect to the measure $\mu_{S}^{\tau}(\cdot)$. Further, define
\[\var_S(f)=\mathbb{E}[\var_S^\tau(f)],\]
where the expectation is over the choice of $\tau$, sampled according to the distribution $\mu_{V\setminus S}$ on $[q]^{V \setminus S}$. As before, we will denote $\var_{\{v\}}(f)$ simply by $\var_{v}(f)$. 

\begin{definition}[Approximate tensorization and block factorization of variance]
We say that the distribution $\mu$ on $[q]^{V}$ satisfies approximate tensorization of variance, with constant $C$, if for all $f:\Omega\to \mathbb{R}_{\ge 0}$, 
\[
\var(f) \le C \sum_{v\in V} \mu[\var_v(f)].
\]
More generally, for $1\leq \ell \leq n$, we say that $\mu$ satisfies $\ell$-uniform block factorization of variance, with constant $C$, if for all $f:\Omega\to \mathbb{R}_{\ge 0}$, 
\begin{equation} \label{ineq:variance factorization}
    \begin{split}
        \frac{\ell}{n} \var(f) \leq \frac{C}{\binom{n}{\ell}}\sum_{S \in \binom{V}{\ell}} \mu[\var_S (f)].
    \end{split}
\end{equation}
\end{definition}

The following assertion, which follows by writing the transition matrix of the Glauber dynamics as the average of the matrices for updating the value at each vertex, provides an immediate connection between approximate tensorization of variance and the spectral gap. 

\begin{fact}[cf.~{\cite[Fact~A.3]{CLV20}}]
\label{claim:spectral-gap-tensorization}
A distribution $\mu$ on $[q]^V$ satisfies approximate tensorization of variance with constant $C$ if and only if the spectral gap of the Glauber dynamics for $\mu$ is at least $\frac{1}{C|V|}$.
\end{fact}

%In particular, the following is an immediate consequence of \cref{thm:ALO}.

%\begin{corollary}
%\label{lem: spectral gap}
%If $\mu$ is $(C,\eta)$-spectrally independent for $C \geq 0$ and $\eta \in [0,1)$, then for all $\emptyset \neq U \subseteq V$, $\Lambda = V \setminus U$, every boundary condition $\tau \in \Omega_{\Lambda}$, and for every $f : \Omega_{U}^{\tau} \to \mb{R}_{\geq 0}$,
%\[\var_{U}^{\tau}(f) \leq \frac{|U|^{2C}}{(1-\eta)^{1+2C}} \sum_{u \in U}\mu^{\tau}_{U}[\var_{u}^{\tau}(f)].\]

%\end{corollary}

%variance block factorization for block of size $\ell = \lceil \theta n \rceil$ for some constant $C = (2/\theta)^{O(\eta)}$ *(for any $\theta \in (0,1)$% (TODO)

\subsection{Block factorization of variance for spectrally independent distributions} In order to prove \cref{thm:spectral gap}, it suffices to prove the corresponding result for approximate tensorization of variance. Following \cite{CLV20}, we will do this in two steps. First, we will show that a $(C,\eta)$-spectrally independent distribution satisfies $\ell$-uniform block factorization with constant $C_{\on{BF}}$, for $\ell = \theta n$ (where $\theta$ is a suitable constant depending on $C$) and $C_{\on{BF}} \leq (2/\theta)^{\lceil 2C\rceil}$. In the next subsection, we will show how to translate such a bound on the $\ell$-uniform block factorization to a bound on the approximate tensorization of variance. Compared to \cite{CLV20}, the crucial difference in our work is that we do not require any lower bound on the marginals of various distributions, instead using the (weak) bound of \cref{lem: spectral gap}.\\ 

Before stating the main result of this subsection, we need to introduce some further notation. Given a set $X$ with $|X|\geq n$, a distribution $\mu$ over $\binom{X}{n}$, and an integer $1\leq s\leq n$, we define the distribution $\mu^{(s)}$ on $\binom{X}{s}$ by 
\[
\mu^{(s)}(S)=\frac{1}{\binom{n}{s}}\sum_{S'\in \binom{X}{n},S\subseteq S'} \mu(S') \quad \text{ for all }S \in \binom{X}{s}.\]
Let $r\le s\le n$. The $\binom{X}{s} \times \binom{X}{r}$ matrix $D_{s,r}$ is the transition matrix corresponding to moving from a given $S \in \binom{X}{s}$ to a uniformly random $r$-subset of $S$. The $\binom{X}{r} \times \binom{X}{s}$ matrix $U_{r,s}$ is the transition matrix corresponding to moving from a given $R \in \binom{X}{r}$ to an $s$-subset containing $R$, such that the probability of moving to any $S \in \binom{X}{s}$ with $S \supseteq R$ is proportional to $\mu^{(s)}(S)$. Note that $U_{r,s}\circ D_{s,r}$ is simply the transition matrix corresponding to the $s \leftrightarrow r$ down-up random walk. We define the $r \leftrightarrow s$ up-down random walk to be the random walk on $\binom{X}{r}$ corresponding to the transition matrix $D_{s,r}\circ U_{r,s}$.

For a function $f^{(s)}:\binom{X}{s}\to \mathbb{R}$, we define $f^{(r)}:\binom{X}{r}\to\mathbb{R}$ by 
\[
f^{(r)}(R) = (U_{r,s}f^{(s)})(R) = \sum_{S\supseteq R,|S|=s} U_{r,s}(R,S)f^{(s)}(S) \quad \text{ for all }R\in \binom{X}{r}
\]

% Our goal in this subsection is to prove the following result showing block factorization of variance, following \cite{CLV20}. In particular, we will show that $\eta$-spectrally independent distributions satisfy $\ell$-block factorization of variance with constant $C=O_{\eta,\Delta}(1)$ for linear block size $\ell = \theta n$ for an appropriate constant $\eta$ depending only on $\eta,\Delta$. In the next subsection, we show that this implies approximate tensorization of variance with constant $C=O_{\eta,\Delta}(1)$. Note that we do not need any condition on the marginals throughout the argument. 

% The result applies more generally to simplicial complices and we will work in that setting. 
\begin{proposition}[cf.~{\cite[Theorem~A.9]{CLV20}}]
\label{lem:l-uniform-factorization}
Let $X$ be a set with $|X|\geq n$. Let $r\leq s \leq n$ be integers and let 
$\mu$ be a $(C,\eta)$-spectrally independent distribution on $\binom{X}{n}$. Then the spectral gap of the $s\leftrightarrow r$ down-up walk and the $r \leftrightarrow s$ up-down walk are bounded below by 
\[
\kappa_{r,s}=\frac{\sum_{k=r}^{s-1}\alpha_0\dots\alpha_{k-1}}{\sum_{k=0}^{s-1}\alpha_0\dots\alpha_{k-1}}, 
\]
where $$\alpha_i = \frac{1-\min(\eta,C/(n-i-1))}{1+\min(\eta,C/(n-i-1))}.$$ 

In particular, for $ \theta n \geq 4\cdot \lceil 2 C \rceil$, $\mu$ satisfies $\lceil \theta n \rceil $-uniform block factorization of variance with constant 
\[
C_{\lceil \theta n \rceil } \leq \left(\frac{2}{\theta}\right)^{\lceil 2C\rceil + 1}.
\]
%assuming that $\theta n \ge \lceil 2C\rceil$.
\end{proposition}

It is immediate from \cref{lem:l-uniform-factorization} (and also alternatively by combining \cref{thm:ALO} with \cref{claim:spectral-gap-tensorization}) that the following holds. 
\begin{corollary}
\label{lem: spectral gap}
If $\mu$ is $(C,\eta)$-spectrally independent for $C \geq 0$ and $\eta \in [0,1)$, then for all $\emptyset \neq U \subseteq V$, $\Lambda = V \setminus U$, every boundary condition $\tau \in \Omega_{\Lambda}$, and for every $f : \Omega_{U}^{\tau} \to \mb{R}_{\geq 0}$,
\[\var_{U}^{\tau}(f) \leq \frac{|U|^{2C}}{(1-\eta)^{1+2C}} \sum_{u \in U}\mu^{\tau}_{U}[\var_{u}^{\tau}(f)].\]
\end{corollary}

For the reader's convenience, we provide the proof of \cref{lem:l-uniform-factorization}, closely following \cite{CLV20}, in \cref{sec:proof-factorization}.

\subsection{Proof of \cref{thm:spectral gap}}

Finally, we show how to convert a bound for uniform $\ell$-block factorization of variance to a bound for uniform $1$-block factorization of variance, provided that the distribution is spectrally independent. As mentioned earlier, in contrast to \cite[Lemma~A.4]{CLV20}, we do not require any lower bound on the marginals of $\mu$.

We need the following standard graph-theoretic input. For a subset $S \subseteq V$, let $C(S)$ denote the set of connected components of the induced graph $G[S]$. We denote by $S_v$ the (unique) element of $C(S)$ containing $v$ i.e.~the connected component of $G[S]$ containing $v$.

\begin{fact}[cf.~{\cite[Lemma~4.3]{CLV20}}]
\label{fact:number-connected}
Let $G = (V,E)$ be a graph with maximum degree $\Delta$ and let $v \in V$. Then, for every integer $k \geq 1$,
\[\mb{P}_{S}[|S_v| = k] \leq \frac{\ell}{|V|}\cdot (2e\Delta \theta)^{k-1},\]
where the probability is taken over a uniformly random subset $S \subseteq V$ of size $\ell = \lceil \theta |V| \rceil$.
\end{fact}

We will also need the following fact regarding the factorization of variance for product measures. 

\begin{fact}[cf.~{\cite[Eq.~4]{martinelli2003ising}}]
\label{fact:tensorization-product}
Let $G = (V,E)$ be a graph and let $\mu$ be a distribution on $[q]^{V}$. For every subset $S \subseteq V$, every boundary condition $\tau \in \Omega_{V \setminus S}$, and every function $f \colon \Omega_{S}^{\tau} \to \mb{R}_{\geq 0}$, we have
\[\on{Var}_{S}^{\tau}(f) \leq \sum_{U \in C(S)} \mu_{S}^{\tau}[\on{Var}_{U}(f)] \]
\end{fact}

We can now deduce the relationship between $\ell$-uniform block factorization and approximate tensorization of variance. 

\begin{proposition} 
\label{lem:variance tensor}
Consider a graph $G = (V,E)$ with maximum degree $\Delta$. Suppose there exist $C \geq 0$, $\eta \in [0,1)$ and $C_{\ell}\geq 0$ such that the distribution $\mu$ on $[q]^{V}$ is $(C,\eta)$-spectrally independent and satisfies $\ell$-uniform block factorization of variance with constant $C_\ell$, for $\ell = \lceil \theta |V| \rceil $ with $0 < \theta \leq \frac{1}{4e\Delta}$. Then, $\mu$ satisfies approximate tensorization of variance with constant $$C_1 = C_{\ref{lem:variance tensor}}\cdot \frac{C_\ell}{(1-\eta)^{1+2C}} \cdot (2C+2)^{4C+4},$$
where $C_{\ref{lem:variance tensor}}$ is a universal constant. 
\end{proposition}
% An immediate corollary of \cref{lem:variance tensor} is that the spectral gap of Glauber dynamics is $\frac{1}{C_1 n}  = \frac{1}{C O(\eta^{\eta}) n }$. 

\begin{proof}
Let $n = |V|$. Then,  
\begin{align*}
    \var(f) &\leq C_\ell\cdot \frac{n}{\ell} \cdot \frac{1}{\binom{n}{\ell}} \sum_{S \in \binom{V}{\ell}} \mu[ \var_S (f) ] & \textrm{$\ell$-uniform block factorization of variance}\\
    &\leq   C_\ell \cdot \frac{n}{\ell} \cdot \frac{1}{\binom{n}{\ell}}   \sum_{S \in \binom{V}{\ell}} \sum_{U \in C(S) }  \mu[ \var_U (f)] & \textrm{by \cref{fact:tensorization-product}} \\ %p. 48 clv
    &\leq  C_\ell \cdot \frac{n}{\ell} \cdot \frac{1}{\binom{n}{\ell}}   \sum_{S \in \binom{V}{\ell}} \sum_{U \in C(S) }  \frac{|U|^{2C}}{(1-\eta)^{1+2C}} \sum_{v \in U} \mu [\var_v(f)] &\textrm{by \cref{lem: spectral gap}}\\
    &= \frac{C_\ell}{(1-\eta)^{1+2C}} \cdot  \frac{n}{\ell} \cdot  \sum_{v \in V} \mu [\var_v(f)] \sum_{k=1}^{\ell}\PX{S}{|S_v| = k} \cdot k^{2C}\\
    &\leq \frac{C_\ell}{(1-\eta)^{1+2C}}\sum_{v \in V} \mu [\var_v(f)] \sum_{k=1}^{\ell}   k^{2C} (2 e \Delta \theta)^{k-1} &\textrm{by \cref{fact:number-connected}}\\%CLV20, Lemma 4.3, bounded size of S_v
    &\leq \frac{C_\ell}{(1-\eta)^{1+2C}} \sum_{k=1}^{\ell} \frac{ k^{2C}}{2^{k-1}} \sum_{v \in V} \mu [\var_v(f)] &\textrm{using that $\theta \le 1/(4e\Delta)$}\\ 
    &\leq C_{\ref{lem:variance tensor}}\cdot \frac{C_\ell}{(1-\eta)^{1+2C}} (2C+2)^{4C+4} \sum_{v \in V} \mu [\var_v(f)]. & \textrm \qedhere 
\end{align*}
\end{proof}
%\intertext{$\sum_{k=1}^{\infty} \frac{ k^{O(\eta)}}{2^{k-1}} \leq D_{\eta}$ for constant $D_{\eta}$}. 
    %\url{https://math.stackexchange.com/questions/1958327/calculation-of-an-infinite-sum-sum-k-1-infty-frack32k-26}}
\begin{proof}[Proof of \cref{thm:spectral gap}]
This follows immediately by combining \cref{lem:l-uniform-factorization} with \cref{lem:variance tensor} and \cref{claim:spectral-gap-tensorization}.
\end{proof}

\section{Applications}

\subsection{Proof of \cref{thm:coloring}}
By making the implicit dependence on $\delta$ sufficiently large, we may assume that $\delta^{-4} = O(\log\log{n})$. We may assume that $k \leq 2\Delta$, since for $k > 2\Delta$, optimal mixing of the Glauber dynamics on any graph is already well-known \cite{jerrum1995very}. We may also assume that $\Delta  = O(\log{n}/\delta^{2})$, since by \cite[Theorem~1.4]{HV06}, for 
$$k \geq \max\set{(1+\delta) \alpha^* \Delta, 288 \ln (96n^3/\delta)/\delta^2},$$
the mixing time of the Glauber dynamics is $O(n\log{n}/\delta)$ (in particular, by \cref{eqn:mixing-to-gap}, the spectral gap is $\tilde{O}_{\delta}(1/n)$).

By \cite[Lemma~6.4]{feng2021rapid} (see also \cite{chen2021rapid}), the uniform distribution $\mu$ on $k$-colorings of $G$ is $(C,\eta)$-spectrally independent with 
\[C = O(\delta^{-2}), \quad \eta = 1-k^{-O(\delta^{-2})}\]
Substituting this into \cref{thm:spectral gap} shows that the spectral gap of the Glauber dynamics is at least
\[\Omega\left(\frac{1}{n}\cdot \frac{1}{\Delta^{O(1/\delta^2)}k^{O(1/\delta^4)}}\right) = \Omega \left(\frac{1}{n}\cdot \frac{1}{(\log{n})^{O(\delta^{-4})}}\right) = \tilde{\Omega}_{\delta}\left(\frac{1}{n}\right),\]
where the first bound uses $k \leq 2\Delta$, $\Delta = O(\log{n}/\delta^2)$, $\delta^{-4} = O(\log\log{n})$ and the second bound uses $\delta^{-4} = O(\log\log{n})$.

\subsection{Proof of \cref{thm:hardcore}}
We begin by bounding the mixing time from a warm start. By making the implicit dependence on $\delta$ sufficiently large, we may assume that $\delta^{-2} = O(\log\log{n})$. Further, we may assume that $\Delta = O(\log(n/\delta))$ since for $\Delta = \Omega(\log(n/\delta))$, \cite[Theorem~4.1]{HV06} already gives mixing time of $O(n\log{n}/\delta)$ from a warm-start. By \cite[Theorem~7]{chen2020rapid}, the Gibbs distribution is $(C,\eta)$-spectrally independent with
\[C = O(\delta^{-1}), \quad \eta = \frac{\lambda}{1+\lambda}.\]
Substituting this into \cref{thm:spectral gap}, we see that the spectral gap is at least
\[\Omega\left(\frac{1}{n}\cdot \frac{1}{\Delta^{O(\delta^{-1})}}\right) = \tilde{\Omega}_{\delta}\left(\frac{1}{n}\right),\]
where the second term follows from the first using the assumptions on $\delta$ and $\Delta$. Therefore, by \cref{eqn:gap-to-warm}, the mixing time of the Glauber dynamics from a warm start is 
\[\tilde{O}_{\delta}(n),\]
where the tilde conceals a factor of $e^{O((\log\log n)^2)}$. 

Therefore, by the adaptive simulated annealing algorithm of \cite[Theorem~7.5,~Corollary~7.6]{vstefankovivc2009adaptive}, there is an algorithm which, given $\varepsilon > 0$ outputs with constant probability a $(1+\varepsilon)$-multiplicative approximation of the partition function in time
\[\tilde{O}_{\delta}(n) \times n \times \on{poly}(\log{n})\times \varepsilon^{-2}\log(\varepsilon^{-1}),\]
as desired. 

\subsection{Proof of \cref{thm:matching}}
By \cite[Theorem 6.1]{CLV20}, the monomer-dimer model with $\lambda=1$ over a graph with maximum degree $\Delta \ge 2$ is $(C, \eta)$-spectrally independeny with $C = 2\sqrt{1+\Delta}$ and $\eta = \frac{\lambda}{1+\lambda}=1/2$. Substituting this into \cref{thm:spectral gap} shows that the spectral gap of the Glauber dynamics is at least
\[\Omega\left(\frac{1}{|E|}\cdot \frac{1}{\Delta^{O(\sqrt{1+\Delta})}}\right),\]
which is
\[\Omega\left(\frac{1}{m}\cdot \frac{1}{n^{o(1)}}\right)\]
for 
$$\Delta = o\left(\frac{(\log{n})^{2}}{(\log\log{n})^{2}}\right).$$

\appendix 
\section{Proof of \cref{lem:l-uniform-factorization}}
\label{sec:proof-factorization}
We will use the following notions of variance contraction, which are intimately connected with the spectral gap of the down-up walk. 

\begin{definition}
For integers $r\leq s \leq n \leq |X|$, we say that a distribution $\mu$ on $\binom{X}{n}$ satisfies order-$(r,s)$ global variance contraction with rate $\kappa$ if for all $f^{(s)}:\binom{X}{s} \to \mathbb{R}$, we have 
\[
\var_{\mu^{(r)}}(f^{(r)}) \le (1-\kappa)\var_{\mu^{(s)}}(f^{(s)}).
\]
\end{definition}

The following relates global variance contraction to the spectral gap of the relevant down-up and up-down walks. 

\begin{fact}[cf.~{\cite[Fact~A.6]{CLV20}}]
$\mu$ satisfies order-$(r,s)$ global variance contraction with rate $\kappa$ if and only if the spectral gap of $s\leftrightarrow r$ down-up walk and the $r \leftrightarrow s$ up-down walk are at least $\kappa$. 
\end{fact}

To bound the rate of global variance contraction, we will introduce the notions of local variance contraction and local spectral expansion, the latter of which is controlled by spectral independence. As before, for a distribution $\mu$ on $\binom{X}{n}$ and $\tau \in \textrm{supp}(\mu^{(k)})$, we denote by $\mu_\tau$ the distribution on the $(n-k)$-subsets $S \subseteq X$ which are disjoint from $\tau$ such that for any such subset $S$. the probability $\mu_\tau (S)$ is proportional to $\mu(\tau \cup S)$. 

\begin{definition}[Local variance contraction]
For $n \leq |X|$, we say that a distribution $\mu$ on $\binom{X}{n}$ satisfies $(\alpha_0,\dots,\alpha_{n-2})$-local variance contraction if the following holds. For every $k\le n-2$ and for every $\tau \in \textrm{supp}(\mu^{(k)})$, $\mu_\tau$ satisfies order-$(1,2)$ global variance contraction with rate $\alpha_k/(1+\alpha_k)$, i.e.
\[
\var_{(\mu_\tau)^{(1)}}(f^{(1)}) \le \frac{1}{(1+\alpha_k)}\var_{(\mu_\tau)^{(2)}}(f^{(2)}).
\]
\end{definition}

\begin{definition}[Local spectral expansion]
For $n\leq |X|$, we say that a distribution $\mu$ on $\binom{X}{n}$ satisfies $(\zeta_0,\dots,\zeta_{n-2})$-local spectral expansion if the following holds. For every $k\le n-2$ and for every $\tau \in \textrm{supp}(\mu^{(k)})$, the second largest eigenvalue of the non-lazy chain on $X\setminus \tau$ induced by the up-down walk on $\mu_\tau^{(1)}$ is at most $\zeta_{k}$.
\end{definition}
Spectral independence implies local spectral expansion. 
\begin{fact}[cf.~{\cite[Theorem~8]{chen2021rapid}}]
\label{fact:si-implies-lse}
For $n\leq |X|$, let $\mu$ be a $(C,\eta)$-spectrally independent distribution on $\binom{X}{n}$. Then $\mu$ satisfies $(\zeta_0,\dots,\zeta_{n-2})$-local spectral expansion with $$\zeta_i = \min\left(\eta, \frac{C}{n-i-1}\right).$$ 
\end{fact}

Moreover, local variance contraction is equivalent to local spectral expansion, and thus (by \cref{fact:si-implies-lse}), can be deduced from spectral independence.

\begin{fact}[cf.~{\cite[Fact~A.8]{CLV20}}]
\label{fact:lse-lvc}
For $n\leq |X|$, a distribution $\mu$ on $\binom{X}{n}$ satisfies $(\alpha_0,\dots,\alpha_{n-2})$-local variance contraction if and only if $\mu$ satisfies $(\zeta_0,\dots,\zeta_{n-2})$-local spectral expansion with $\zeta_k=\frac{1-\alpha_k}{1+\alpha_k}$. 
\end{fact}

We are now in a position to prove \cref{lem:l-uniform-factorization}.

\begin{proof}[Proof of \cref{lem:l-uniform-factorization}]
Consider $X$ with $|X|\leq n$ and fix integers $r\leq s \leq n$. Suppose $\mu$ is a $(C,\eta)$-spectrally independent distribution on $\binom{X}{n}$. By combining \cref{fact:lse-lvc} and \cref{fact:si-implies-lse}, we have that $\mu$ satisfies $(\alpha_0,\dots,\alpha_{n-2})$-local variance contraction with 
\[
\alpha_i = \frac{1-\min(\eta,C/(n-i-1))}{1+\min(\eta,C/(n-i-1)}.
\]

We show order-$(r,s)$ global variance contraction with constant $\kappa_{r,s}$ as in the statement of the proposition. Let $(X_1,\dots, X_n) \in X^{n}$ be distributed as a uniformly random permutation of an element in $\binom{X}{n}$ sampled according to $\mu$. Note that for any $k \leq n$ and for any $K \in \binom{X}{k}$,
\[\mb{P}[\{X_1,\dots, X_k\} = K] = \mu^{(k)}(K).\]
Let $f^{(s)}:\binom{X}{s}\to \mathbb{R}$. %For a given set $S$ with $s$ elements, consider a random permutation $X_1,\dots,X_s$ of $S$. 
Then, we can decompose $\var_{\mu^{(s)}}(f^{(s)})$ as 
\begin{align*}
&\var_{\mu^{(s)}}(f^{(s)}) \\
&= \mathbb{E}[f^{(s)}(X_1,\dots,X_s)^2]-\mathbb{E}[f^{(s)}(X_1,\dots,X_s)]^2\\
&=\sum_{j=1}^{s} \mathbb{E}\left[\mathbb{E}[f^{(s)}(X_1,\dots,X_s)|X_1,\dots,X_j]^2 - \mathbb{E}[f^{(s)}(X_1,\dots,X_s)|X_1,\dots,X_{j-1}]^2\right]\\
&= \sum_{j=1}^{s}\Delta_j,
\end{align*}
where
 $$\Delta_j = \mathbb{E}\left[\mathbb{E}[f^{(s)}(X_1,\dots,X_s)|X_1,\dots,X_j]^2 - \mathbb{E}[f^{(s)}(X_1,\dots,X_s)|X_1,\dots,X_{j-1}]^2\right].$$
Similarly, for $r\leq s$, we have
\[\var_{\mu^{(r)}}(f^{(s)}) = \sum_{j=1}^{r}\Delta_{j}.\]

%Consider the $j$-th term conditioned on $X_1,\dots,X_{j-2}$. 

Consider the random subset $\tau = \{X_1,\dots,X_{j-2}\}$ and denote by $f_\tau^{(r)}$ the random function induced by $f^{(s)}$ on subsets of size $r$ of $X\setminus \tau$. Concretely, for any subset $R \subseteq X \setminus \tau$ with $|R| = r$, we have
\[f_{\tau}^{(r)}(R) = f^{(r+j-2)}(R \cup \tau).\]
Then,  
\begin{align*}
&\mathbb{E}\left[\mathbb{E}[f^{(s)}(X_1,\dots,X_s)|X_1,\dots,X_j]^2 - \mathbb{E}[f^{(s)}(X_1,\dots,X_s)|X_1,\dots,X_{j-2}]^2\mid X_1,\dots,X_{j-2}\right] \\
&= \mb{E}\left[\var_{\mu_{\tau}^{(2)}}(f_\tau^{(2)})\right],
\end{align*}
and
\begin{align*}
&\mathbb{E}\left[\mathbb{E}[f^{(s)}(X_1,\dots,X_s)|X_1,\dots,X_{j-1}]^2 - \mathbb{E}[f^{(s)}(X_1,\dots,X_s)|X_1,\dots,X_{j-2}]^2\mid X_1,\dots,X_{j-2}\right] \\
&= \mb{E}\left[\var_{\mu_{\tau}^{(1)}}(f_\tau^{(1)})\right].
\end{align*}
For brevity of notation, let $$A_j = \mathbb{E}[f^{(s)}(X_1,\dots,X_s)|X_1,\dots,X_j]^2.$$
By the assumption of local variance contraction, we have for every realisation of $\tau$ that
\[\var_{\mu_{\tau}^{(1)}}(f_\tau^{(1)}) \leq \frac{1}{(1+\alpha_{j-2})} \var_{\mu_{\tau}^{(2)}}(f_\tau^{(2)}).\]
Therefore, from the above identities, we have that
% and $$\Delta_j = \mathbb{E}[\mathbb{E}[f^{(s)}(X_1,\dots,X_s)|X_1,\dots,X_j]^2 - \mathbb{E}[f^{(s)}(X_1,\dots,X_s)|X_1,\dots,X_{j-1}]^2].$$
%By the local variance contraction assumption, we have 
\[
(1+\alpha_{j-2})\mathbb{E}[A_{j-1}-A_{j-2}|X_1,\dots,X_{j-2}] \le \mathbb{E}[A_{j}-A_{j-2}|X_1,\dots,X_{j-2}].
\]
Taking the expectation over $X_1,\dots, X_{j-2}$, we can conclude that
\[
\Delta_{j} \ge \alpha_{j-2} \Delta_{j-1},
\]
i.e.~for all $2\leq j \leq s$,
\[\var_{\mu^{(j)}}(f^{(s)}) \geq (1+\alpha_{j-2})\var_{\mu^{(j-1)}}(f^{(s)}) - \alpha_{j-2}\var_{\mu^{(j-2)}}(f^{(s)}).\]
Using this inductively (cf.~\cite[Proof of Theorem~5.4]{CLV20}), we obtain that 
\begin{align*}
\frac{\var_{\mu^{(s)}}(f^{(s)})}{\var_{\mu^{(r)}}(f^{(s)})} 
&\ge \frac{\sum_{j=1}^{s} \prod_{i=0}^{j-2}\alpha_i}{\sum_{j=1}^{r}\prod_{i=0}^{j-2}\alpha_i}.
\end{align*}
Thus, by definition, we have order-$(r,s)$ global variance contraction with rate 
\[
\kappa_{r,s} = \frac{\sum_{j=r+1}^{s} \prod_{i=0}^{j-2}\alpha_i}{\sum_{j=1}^{s}\prod_{i=0}^{j-2}\alpha_i}.
\]
Let $f \colon \binom{X}{n} \to \mb{R}$. Using the identity
\[\on{Var}_{\mu}(f) - \on{Var}_{\mu^{(n - \ell)}}(f^{(n-\ell)}) = \frac{1}{\binom{n}{\ell}}\sum_{S \in \binom{X}{\ell}}\mu[\on{Var}_S(f)],\]
it follows from above that $\mu$ has $\ell$-uniform block factorization of variance with constant at most
\[\frac{\ell}{n}\cdot \frac{1}{\kappa_{n-\ell, n}}.\]
% \[
% \frac{\sum_{k=0}^{n-1}\alpha_0\dots\alpha_{k-1}}{\sum_{k=n-\theta n}^{n-1}\alpha_0\dots\alpha_{k-1}}.
% \]
Since $\alpha_k \ge \max\left(1-\frac{\lceil 2C\rceil}{n-k-1},0\right)$ and since ${\kappa_{n-\ell,n}}$ is monotone increasing in each $\alpha_k$, it follows that for $\ell = \lceil \theta n \rceil$ with $\theta n \geq 4\cdot \lceil 2 C\rceil$, we have 
\begin{align*}
   {\kappa_{n-\ell,n}} = \frac{\sum_{k=n-\ell+1}^{n}\alpha_0\dots\alpha_{k-2}}{\sum_{k=1}^{n}\alpha_0\dots\alpha_{k-2}} 
    &\ge \frac{\sum_{k=n-\ell+1}^{n}(n-k)\dots(n-k + 1 -\lceil2C\rceil)}{\sum_{k=1}^{n}(n-k)\dots(n-k+1-\lceil2C\rceil)} \\
    &\ge \frac{(\theta n/2)\cdot (\theta n/2-1)\dots(\theta n/2-\lceil2C\rceil)}{n\cdot (n-1)\dots(n-\lceil 2C\rceil)}\\
    &\ge \frac{\theta}{2}\cdot (\theta/2)^{\lceil 2C\rceil},
\end{align*}
which gives the desired assertion.
\end{proof}
\bibliographystyle{alpha}
\bibliography{main.bib}
\end{document}

%% file: abstract.tex
 We present a new lower bound on the spectral gap of the Glauber dynamics for the Gibbs distribution of a spectrally independent $q$-spin system on a graph $G = (V,E)$ with maximum degree $\Delta$. Notably, for several interesting examples, our bound covers the entire regime of $\Delta$ excluded by arguments based on coupling with the stationary distribution. As concrete applications, by combining our new lower bound with known spectral independence computations and known coupling arguments:
    \begin{itemize}
        \item We show that for a triangle-free graph $G = (V,E)$ with maximum degree $\Delta \geq 3$, the Glauber dynamics for the uniform distribution on proper $k$-colorings with $k \geq (1.763\dots + \delta)\Delta$ colors has spectral gap $\tilde{\Omega}_{\delta}(|V|^{-1})$. Previously, such a result was known either if the girth of $G$ is at least $5$ [Dyer et.~al, FOCS 2004], or under restrictions on $\Delta$ [Chen et.~al, STOC 2021; Hayes-Vigoda, FOCS 2003].
        \item We show that for a regular graph $G = (V,E)$ with degree $\Delta \geq 3$ and girth at least $6$, and for any $\varepsilon, \delta > 0$, the partition function of the hardcore model with fugacity $\lambda \leq (1-\delta)\lambda_{c}(\Delta)$ may be approximated within a $(1+\varepsilon)$-multiplicative factor in time $\tilde{O}_{\delta}(n^{2}\varepsilon^{-2})$. Previously, such a result was known if the girth is at least $7$ [Efthymiou et.~al, SICOMP 2019].
        \item We show for the binomial random graph $G(n,d/n)$ with $d = O(1)$, with high probability, an approximately uniformly random matching may be sampled in time $O_{d}(n^{2+o(1)})$. This improves the corresponding running time of $\tilde{O}_{d}(n^{3})$ due to [Jerrum-Sinclair, SICOMP 1989; Jerrum, 2003].
    \end{itemize}

%% file: main.bbl
\newcommand{\etalchar}[1]{$^{#1}$}
\begin{thebibliography}{ALOG20}

\bibitem[AF02]{aldous2002reversible}
David Aldous and Jim Fill.
\newblock Reversible {M}arkov chains and random walks on graphs, 2002.

\bibitem[AL20]{alev2020improved}
Vedat~Levi Alev and Lap~Chi Lau.
\newblock Improved analysis of higher order random walks and applications.
\newblock In {\em Proceedings of the 52nd Annual ACM SIGACT Symposium on Theory
  of Computing}, pages 1198--1211, 2020.

\bibitem[ALOG20]{ALO20}
Nima Anari, Kuikui Liu, and Shayan Oveis~Gharan.
\newblock Spectral independence in high-dimensional expanders and applications
  to the hardcore model.
\newblock In {\em Proceedings of the 61st {IEEE} Annual Symposium on
  Foundations of Computer Science}. {IEEE} Computer Society, November 2020.

\bibitem[BCC{\etalchar{+}}21]{blanca2021mixing}
Antonio Blanca, Pietro Caputo, Zongchen Chen, Daniel Parisi, Daniel
  {\v{S}}tefankovi{\v{c}}, and Eric Vigoda.
\newblock On mixing of {M}arkov chains: Coupling, spectral independence, and
  entropy factorization.
\newblock {\em arXiv preprint arXiv:2103.07459}, 2021.

\bibitem[CG{\v{S}}V21]{chen2021rapid}
Zongchen Chen, Andreas Galanis, Daniel {\v{S}}tefankovi{\v{c}}, and Eric
  Vigoda.
\newblock Rapid mixing for colorings via spectral independence.
\newblock In {\em Proceedings of the 2021 ACM-SIAM Symposium on Discrete
  Algorithms (SODA)}, pages 1548--1557. SIAM, 2021.

\bibitem[CLV20a]{CLV20}
Zongchen Chen, Kuikui Liu, and Eric Vigoda.
\newblock Optimal mixing of {G}lauber dynamics: Entropy factorization via
  high-dimensional expansion.
\newblock {\em arXiv preprint arXiv:2011.02075}, 2020.

\bibitem[CLV20b]{chen2020rapid}
Zongchen Chen, Kuikui Liu, and Eric Vigoda.
\newblock Rapid mixing of {G}lauber dynamics up to uniqueness via contraction.
\newblock {\em arXiv preprint arXiv:2004.09083}, 2020.

\bibitem[DFHV04]{dyer2004randomly}
Martin Dyer, Alan Frieze, Thomas~P Hayes, and Eric Vigoda.
\newblock Randomly coloring constant degree graphs.
\newblock In {\em 45th Annual IEEE Symposium on Foundations of Computer
  Science}, pages 582--589. IEEE, 2004.

\bibitem[DK17]{dinur2017high}
Irit Dinur and Tali Kaufman.
\newblock High dimensional expanders imply agreement expanders.
\newblock In {\em 2017 IEEE 58th Annual Symposium on Foundations of Computer
  Science (FOCS)}, pages 974--985. IEEE, 2017.

\bibitem[EHS{\etalchar{+}}19]{efthymiou2019convergence}
Charilaos Efthymiou, Thomas~P Hayes, Daniel Stefankovic, Eric Vigoda, and
  Yitong Yin.
\newblock Convergence of {MCMC} and loopy {BP} in the tree uniqueness region
  for the hard-core model.
\newblock {\em SIAM Journal on Computing}, 48(2):581--643, 2019.

\bibitem[FGYZ21]{feng2021rapid}
Weiming Feng, Heng Guo, Yitong Yin, and Chihao Zhang.
\newblock Rapid mixing from spectral independence beyond the {B}oolean domain.
\newblock In {\em Proceedings of the 2021 ACM-SIAM Symposium on Discrete
  Algorithms (SODA)}, pages 1558--1577. SIAM, 2021.

\bibitem[Hay13]{hayes2013local}
Thomas~P Hayes.
\newblock Local uniformity properties for {G}lauber dynamics on graph
  colorings.
\newblock {\em Random Structures \& Algorithms}, 43(2):139--180, 2013.

\bibitem[HS05]{hayes2005general}
Thomas~P Hayes and Alistair Sinclair.
\newblock A general lower bound for mixing of single-site dynamics on graphs.
\newblock In {\em 46th Annual IEEE Symposium on Foundations of Computer Science
  (FOCS'05)}, pages 511--520. IEEE, 2005.

\bibitem[HV03]{hayes2003non}
Thomas~P Hayes and Eric Vigoda.
\newblock A non-{M}arkovian coupling for randomly sampling colorings.
\newblock In {\em 44th Annual IEEE Symposium on Foundations of Computer
  Science, 2003. Proceedings.}, pages 618--627. IEEE, 2003.

\bibitem[HV06]{HV06}
Thomas Hayes and Eric Vigoda.
\newblock Coupling with the stationary distribution and improved sampling for
  colorings and independent sets.
\newblock {\em The Annals of Applied Probability}, 16, 11 2006.

\bibitem[Jer95]{jerrum1995very}
Mark Jerrum.
\newblock A very simple algorithm for estimating the number of k-colorings of a
  low-degree graph.
\newblock {\em Random Structures \& Algorithms}, 7(2):157--165, 1995.

\bibitem[Jer03]{jerrum2003counting}
Mark Jerrum.
\newblock {\em Counting, sampling and integrating: algorithms and complexity}.
\newblock Springer Science \& Business Media, 2003.

\bibitem[JS89]{jerrum1989approximating}
Mark Jerrum and Alistair Sinclair.
\newblock Approximating the permanent.
\newblock {\em SIAM journal on computing}, 18(6):1149--1178, 1989.

\bibitem[KO18]{kaufman2018high}
Tali Kaufman and Izhar Oppenheim.
\newblock High order random walks: Beyond spectral gap.
\newblock In {\em Approximation, Randomization, and Combinatorial Optimization.
  Algorithms and Techniques (APPROX/RANDOM 2018)}. Schloss
  Dagstuhl-Leibniz-Zentrum f{\"u}r Informatik, 2018.

\bibitem[LM06]{lau2006randomly}
Lap~Chi Lau and Michael Molloy.
\newblock Randomly colouring graphs with girth five and large maximum degree.
\newblock In {\em Latin American Symposium on Theoretical Informatics}, pages
  665--676. Springer, 2006.

\bibitem[LP17]{LP17}
David Levin and Yuval Peres.
\newblock {\em Markov Chains and Mixing Times}.
\newblock 10 2017.

\bibitem[MSW03]{martinelli2003ising}
Fabio Martinelli, Alistair Sinclair, and Dror Weitz.
\newblock The ising model on trees: Boundary conditions and mixing time.
\newblock In {\em 44th Annual IEEE Symposium on Foundations of Computer
  Science, 2003. Proceedings.}, pages 628--639. IEEE, 2003.

\bibitem[Opp18]{oppenheim2018local}
Izhar Oppenheim.
\newblock Local spectral expansion approach to high dimensional expanders part
  i: Descent of spectral gaps.
\newblock {\em Discrete \& Computational Geometry}, 59(2):293--330, 2018.

\bibitem[SJ89]{sinclair1989approximate}
Alistair Sinclair and Mark Jerrum.
\newblock Approximate counting, uniform generation and rapidly mixing {M}arkov
  chains.
\newblock {\em Information and Computation}, 82(1):93--133, 1989.

\bibitem[{\v{S}}VV09]{vstefankovivc2009adaptive}
Daniel {\v{S}}tefankovi{\v{c}}, Santosh Vempala, and Eric Vigoda.
\newblock Adaptive simulated annealing: A near-optimal connection between
  sampling and counting.
\newblock {\em Journal of the ACM (JACM)}, 56(3):1--36, 2009.

\end{thebibliography}
